\documentclass[twoside,leqno,11pt]{article}

\usepackage[letterpaper]{geometry}

\usepackage{graphicx} 
\usepackage{ltexpprt}
\usepackage{booktabs}
\usepackage{amsfonts}
\usepackage{amssymb}

\usepackage{amstext}
\usepackage{amsmath}
\usepackage{xspace}
\usepackage{theorem}
\usepackage{graphicx}
\usepackage{mdframed}
\usepackage{comment}
\usepackage{mathdots}
\usepackage{xcolor}
\usepackage{graphics}
\usepackage{colordvi}
\usepackage{color}

\usepackage{paralist}
\usepackage{bm}
\usepackage{mathpazo}

\usepackage{float}
\usepackage{kbordermatrix}
\usepackage{thmtools, thm-restate}

\usepackage{titlesec}

\titleformat{\section}
  {\normalfont\Large\bfseries}  
  {\thesection}{1em}{}          

\titleformat{\subsection}
  {\normalfont\large\bfseries}
  {\thesubsection}{1em}{}

\usepackage{tikz}
\usetikzlibrary{snakes}
\usepackage[boxruled,lined,linesnumbered]{algorithm2e}
\usepackage{todonotes}
\usepackage{subcaption}
\usepackage{enumitem}
\usetikzlibrary{calc}
\usepackage{lineno}
\usepackage[colorlinks=true, linkcolor=red!70!white, citecolor=blue!80!white, urlcolor=blue]{hyperref}
\usepackage[capitalise,nameinlink]{cleveref}
\crefname{figure}{Figure}{Figures}

\newtheorem{theorem}{Theorem}

\newtheorem{lemma}{Lemma}
\newtheorem{observation}{Observation}

\newtheorem{corollary}{Corollary}

\newtheorem{definition}{Definition}

\newenvironment{proof}[1]{\noindent{\bf Proof }#1}%
        {\hspace*{\fill}$\Box$\par\vspace{4mm}}

\usepackage{authblk}

\newcommand{\RR}{{\mathbb{R}}}
\newcommand{\cset}{{\mathcal C}}

\newcommand{\dset}{{\mathcal D^2}}

\title{A Simple PTAS for Weighted $k$-means and Sensor Coverage}

\author{
    Akash Pareek\\
    Department of Computer Science and Automation\\
    IISc Bangalore\\
    \texttt{akashpareek@iisc.ac.in}

   \and

    Supratim Shit\\
    IIIT Delhi\\
    \texttt{supratim@iiitd.ac.in}
}
\date{}
\begin{document}

\maketitle







\begin{abstract}
Clustering is a fundamental technique in data analysis, with the $k$-means problem being one of the most widely studied objectives due to its simplicity and broad applicability. In many practical scenarios, data points come with associated weights that reflect their importance, frequency, or confidence. Given a weighted point set $P \subset \mathbb{R}^d$, where each point $p \in P$ has a positive weight $w_p$, the goal is to compute a set of $k$ centers $\mathcal{C} = \{c_1, c_2, \dots, c_k\} \subset \mathbb{R}^d$ that minimizes the weighted clustering cost: $\Delta^w(P, \mathcal{C}) = \sum_{p \in P} w_p \cdot d(p, \mathcal{C})^2,$ where $d(p, \mathcal{C})$ denotes the Euclidean distance from $p$ to its nearest center in $\mathcal{C}$. Although most existing \emph{coreset}-based algorithms for $k$-means extend naturally to the weighted setting and provide a PTAS, no prior work has offered a simple, coreset-free PTAS  designed specifically for the weighted $k$-means problem.

In this paper, we present a simple PTAS for weighted 
$k$-means that do not rely on coresets. Building upon the framework of Jaiswal, Kumar, and Sen \cite{jaiswal2012simple} for the unweighted case, we extend the result to the weighted setting by using the weighted $\dset$-sampling technique. Our algorithm runs in time $nd \cdot 2^{\tilde{O}(k^2/\epsilon)}$ and outputs a set of $k$ centers whose total clustering cost is within a $(1+\epsilon)$-factor of the optimal cost.

As a key application of the weighted $k$-means, we obtain a PTAS for the \emph{sensor coverage problem}, which can also be viewed as a continuous locational optimization problem. For this problem, the best-known result prior to our work was an $O(\log k)$-approximation due to Deshpande \cite{deshpande2014guaranteed}. In contrast, our algorithm guarantees a $(1 + \epsilon)$-approximation to the optimal coverage cost even before applying refinement steps like Lloyd’s descent.

\newpage

\end{abstract}
\section{Introduction}
Clustering is one of the fundamental steps in the analysis of large datasets. A classical and widely used approach to cluster data is the \emph{$k$-means clustering}, where the input is a point set \( P \subset \mathbb{R}^d \) of size $n$, and the objective is to compute a set of \( k \) centers \( \mathcal{C} = \{c_1, c_2, \dots, c_k\} \subset \mathbb{R}^d \) that minimizes the total squared distance of each point to its nearest center, formally given by:

$$\Delta(P, \mathcal{C}) = \sum_{p \in P} d(p, \mathcal{C})^2,$$

where \( d(p, \mathcal{C})^2= \min_{c \in \mathcal{C}} \|p - c\|^2 \) denotes the squared Euclidean distance from point \( p \) to the closest center in \( \mathcal{C} \).

As the objective function of \(k\)-means clustering is both intuitive and computationally well-structured, it has been widely applied across various domains, including machine learning, image processing, dimensionality reduction, and community detection in large-scale graphs.

The $k$-means objective implicitly assumes that all the points in the dataset are equally important. However, in many practical applications, data points may have different importance or influence. This asymmetry can be naturally modeled by associating each point \( p \in P \) with a non-negative real-value weight \( w_p\), representing its significance in the dataset. This leads to the \emph{weighted $k$-means} problem, which generalizes the $k$-means by minimizing the weighted clustering cost:

$$ \Delta^w(P, \mathcal{C}) = \sum_{p \in P} w_p \cdot d(p, \mathcal{C})^2. $$

Here, the goal is to compute \( k \) centers \( \mathcal{C} \subset \mathbb{R}^d \) that minimize \( \Delta^w(P, \mathcal{C}) \). This formulation captures scenarios in which each data point contributes to the clustering objective proportionally to its associated weight. Weighted \(k\)-means arises naturally in a variety of real-world applications, including those involving aggregated observations, non-uniform sampling frequencies, varying confidence levels, geospatial analysis, document and text clustering, imbalanced datasets, and so on.

Despite extensive research on the $k$-means problem, it remains NP-hard to compute an optimal set of centers \( \mathcal{C} \) even for the case \( k = 2 \) \cite{dasgupta2008hardness}. A widely used heuristic for solving $k$-means in practice is \emph{Lloyd’s algorithm} \cite{lloyd1982least}. While Lloyd’s algorithm often performs well empirically, it is known to converge to local minima and may exhibit poor worst-case behavior \cite{arthur2006slow}. To address the sensitivity of Lloyd’s Algorithm, Arthur and Vassilvitskii proposed the \emph{$k$-means++} algorithm \cite{arthur2006k},  that instead of choosing the initial $k$-centers at random (like Lloyd's), picks the initial $k$ centers using $\dset$-sampling. The $\dset$-sampling can be informally defined as follows: Let $\cset'$ be the set of centers chosen by the $\dset$-sampling. Initially, $\cset'$ is empty. The first center is chosen at random and inserted into the set $\cset'$. Any other point $p \in P$ is selected as the next center with probability $d(p,\cset')$. This process is continued till $\cset'$ has exactly $k$ centers. Now, the set $\cset'$ is used as initial centers for Lloyd's algorithm. Using this $\dset$-sampling technique, Arthur and Vassilvitskii showed that the $k$-means++ algorithm exhibits $O(\log k)$-approximation to the optimal clustering.

The $k$-means problem being NP-hard, researchers always try to find good approximate solutions for the problem, i.e., to return a set of centers such that the $k$-means cost is close to the optimal $k$-means cost. Matousek \cite{matouvsek2000approximate} gave a $(1+\epsilon)$-approximation which runs in time $O(n \log k \hspace{.1cm} n \epsilon^{-2k^2d})$. Kanungo et al. \cite{kanungo2002local} gave a $(9+ \epsilon)$-approximation algorithm which runs in time $O(n\log n + n\epsilon^{-d}\log (1/\epsilon)+ n^2k^3\log n)$. Kumar et al. \cite{kumar2004simple} also gave an $(1+\epsilon)$-approximation which runs in time $nd\cdot2^{(k/\epsilon)^{O(1)}}$. Jaiswal et al. \cite{jaiswal2012simple} gave a simple $\dset$-sampling based $(1+\epsilon)$-approximation which runs in time $nd\cdot 2^k\cdot 2^{\tilde{O}(k/\epsilon)}$. There are many other works related to $k$-means, we mention a few of them in \cref{sec:relatedwork}.

In addition to the above results for $k$-means, a lot of work for $k$-means has focused on designing \emph{coreset-based algorithms}. A coreset is a small, weighted subset of the input data that approximates the original data. For the $k$-means objective function, the clustering cost on the coreset closely matches the cost on the full dataset. Coreset-based approaches typically proceed in two steps. First, a coreset is constructed, which is the most technically involved part of the process. Unlike simple sampling techniques, coreset construction often requires careful sensitivity analysis, importance sampling, and geometric partitioning, making it algorithmically and analytically more complicated. The second step applies a possibly exhaustive search or approximation algorithm on the small coreset instead of the full dataset. Since the size of the coreset is usually independent of the size of the original input (and often depends only on parameters like $k$, $d$, and $\epsilon$), this enables efficient PTAS for $k$-means.

All coreset-based algorithms for unweighted $k$-means naturally extend to the weighted setting, thereby yielding PTAS for weighted $k$-means as well. One of the earliest such constructions is due to Har-Peled and Mazumdar~\cite{har2004coresets}, who developed a coreset of size $O(k \cdot \epsilon^{-d} \log n)$ for low-dimensional Euclidean spaces. This construction leads to a PTAS with a running time $O\left(n + k^{k+2} \cdot \epsilon^{-(2d+1)k} \cdot \log^{k+1} n \cdot \log^{k}(1/\epsilon)\right)$. More recently, ~\cite{cohen2022improved} presented improved coreset constructions for $k$-means clustering, achieving a size of $\tilde{O}\left( \min(k^{4/3} \cdot \epsilon^{-2},\ k \cdot \epsilon^{-4}) \right)$. Several other coreset-based techniques have been proposed over the years, each offering different trade-offs between coreset size and construction time, we discuss a few of them in \cref{sec:relatedwork}. To the best of our knowledge, the coreset-based approach is the only method that obtains a PTAS for the weighted $k$-means.

As mentioned earlier, coreset-based methods can be quite complex to analyze and often difficult to implement due to the technical nature of coreset construction. Therefore, in this paper, we present a \emph{simple} $(1 + \epsilon)$-approximation algorithm for the weighted $k$-means problem using weighted $D^2$-sampling. Our method avoids coresets entirely and runs in time $nd \cdot 2^{\tilde{O}(k^2 / \epsilon)}$. Our result extends the algorithm of Jaiswal et al.~\cite{jaiswal2012simple} to the weighted setting. Compared to coreset-based approaches, our algorithm is much simpler to analyze and is easy to implement.

As an application of our result, we next provide a $(1+\epsilon)$-approximation for the \emph{sensor covering problem}. Sensor coverage is a central problem in large-scale systems. At its core, it deals with the fundamental question of \emph{where} to place sensors to effectively monitor a given environment. Over the past two decades, this problem has attracted considerable attention, leading to a variety of different formulations. Despite the differences in these models, the goal remains the same: to formalize a meaningful metric for the \emph{quality of sensing} and to find sensor placements that either exactly or approximately optimize this metric.

In this paper, we focus on a version of the sensor coverage problem that is framed as a \emph{locational optimization problem}, similar to the formulations studied in \cite{cortes2004coverage}. Such problems have been explored extensively in the context of spatial resource allocation. A widely used approach for solving the sensor covering problem is the \emph{Lloyd’s descent algorithm}, which iteratively adjusts sensor positions to locally minimize the coverage cost. While Lloyd’s method is effective in many practical settings, its performance is highly sensitive to the \emph{initial configuration} of the sensors. In most existing works, this initial configuration is chosen uniformly at random, which offers no guarantee on the quality of the solution obtained after convergence. Deshpande \cite{deshpande2014guaranteed} for the first time showed that the sensor covering problem can be posed as a weighted $k$-means problem and obtained an approximation ratio of $O(\log k)$ by an analysis similar to Arthur and Vassilvitskii \cite{arthur2006k} using weighted $\dset$-sampling for initial configuration. His result guarantees an $O(\log k)$-approximation to the coverage costs even before Lloyd's descent is applied. 

For some $\epsilon \in (0,1)$. using our result of $(1+\epsilon)$ for weighted $k$-means and the relation between sensor covering and weighted $k$-means \cite{deshpande2014guaranteed}, we obtain $(1+ \epsilon)$-approximation for sensor covering even before Lloyd's descent is applied. Similar to weighted $k$-means, the sensor covering algorithm runs in time $n\cdot 2^{\tilde{O}(k^2/\epsilon)}$. We discuss more about the sensor covering problem and the result related to it in \cref{sec:sensor}.

\subsection{Our Results}
In this work, we present a simple PTAS for the weighted \(k\)-means problem. Our algorithm is simple and elegant, and is inspired by the work of  Jaiswal et al. \cite{jaiswal2012simple}, who developed a $k$-means PTAS for the unweighted setting. Our primary result is the following theorem:

\begin{theorem}\label{thm:ptaskmeans}
    Let $P \subset \RR^d$ be a point set such that each point $p \in P$ has positive finite weight $w_p\ge 1$. Then there is a set $\cset=(c_1',c_2',\dots,c_k')$ of $k$ centers computed in polynomial time such that

    $$\Delta^w(P,\cset)\le (1+\epsilon)\Delta_k^w(P),$$ where $\Delta_k^w(P)$ is the cost of the optimal weighted $k$-means.
\end{theorem}

To ensure the above theorem, we design an algorithm that uses weighted \(\mathcal{D}^2\)-sampling. In each iteration of our algorithm (\cref{alg:PTASkmeans}), we draw a sample of points from $P$ using weighted \(\mathcal{D}^2\)-sampling, and then carefully select a candidate subset whose weighted centroid serves as a center. Through a carefully maintained invariant and probabilistic analysis adapted from the ideas of \cite{jaiswal2012simple}, we demonstrate that each selected center closely approximates one of the optimal centers. By repeating this process for all \(k\) centers, we obtain a solution whose total clustering cost is within a \((1 + \epsilon)\) factor of the optimum. An obstruction that we overcome to prove \cref{thm:ptaskmeans} is to show that we can convert a weighted point $p$ with weight $w_p$ to an unweighted point set by making $w_p$ copies of $p$ and still be able to sample $w_p$ such points when required. 

Further, as an application of the weighted $k$-means, we show that the sensor covering problem also admits a PTAS. Formally, we show the following theorem:

\begin{theorem}\label{thm:ptassensor}
     Let $\mathcal{C}=(c_1',c_2',\dots,c_k')$ be the $k$ centers for $\Delta^w(X,\cset)$ obtained using \cref{alg:PTASkmeans}. Then
$$H(\cset)\le(1+\epsilon)H(\cset^o).$$
\end{theorem}

Here, $X$ is a weighted point set obtained by transforming the sensor covering problem to weighted points, $H(\cset)$ is the sensor coverage cost when the sensors are placed at locations $(c_1',c_2',\dots,c_k')$ and $H(\cset^o)$ denoted the cost of the optimal sensor coverage. These quantities are formally defined in \cref{sec:sensor}. To establish the above result, we leverage \cref{thm:ptaskmeans}, in combination with the result of Deshpande \cite{deshpande2014guaranteed} that the sensor coverage problem can be formulated as an instance of weighted $k$-means. This connection allows us to directly inherit the approximation guarantees of \cref{alg:PTASkmeans} to the sensor coverage problem.

\subsection{Related works}\label{sec:relatedwork}
A large body of work has focused on designing coreset-based algorithms for the \(k\)-means problem.  

Chen \cite{chen2006k} presented a PTAS that runs in time \(O(ndk + 2^{(k/\epsilon)^{O(1)}}d^2n^{\sigma})\). Feldman et al. \cite{feldman2007ptas} developed another coreset-based PTAS with a runtime of \(O(nkd + d \cdot \text{poly}(k/\epsilon) + 2^{\tilde{O}(k/\epsilon)})\). In \cite{cohen2021new}, the authors give a simple coreset for $k$-means of size $\tilde{O}(k /\epsilon^4)$. In \cite{huang2024optimal},  the authors gave an improved coreset for  $k$-means of size $\tilde{O}(k^{3/2}\epsilon^{-2})$.  Due to the vast literature on \(k\)-means clustering, we refer readers to the comprehensive surveys by Ahmed et al. \cite{ahmed2020k} and Feldman \cite{feldman2020core} and therein references for further details.

In addition, recent advancements in learning-augmented algorithms \cite{mitzenmacher2022algorithms} have motivated new approaches for \(k\)-means. Ergun et al. \cite{ergun2021learning} introduced the first learning-augmented \(k\)-means algorithm with an approximation ratio of \((1 + O(\alpha))\), where \(\alpha\) denotes the label error rate. Nguyen et al. \cite{nguyen2022improved} improved upon this by designing an algorithm with the same approximation guarantee that runs in time \(O(nd \log^3 (m / \alpha))\). Most recently, Huang et al. \cite{huangnew} proposed a faster learning-augmented algorithm achieving \((1 + O(\alpha))\)-approximation with a runtime of \(O(nd) + \tilde{O}(kd/\alpha)\).

For the sensor coverage problem, exact solutions in the centralized model have been studied by Okabe et al. \cite{okabe1997locational} and Cortés et al. \cite{cortes2004coverage}. Cortés et al. \cite{cortes2004coverage} proposed a distributed Lloyd gradient descent algorithm that guarantees convergence to a centroidal Voronoi configuration. Schwager et al. \cite{schwager2006distributed} extended this framework to handle generalized monotonic sensing functions, while Deshpande et al. \cite{deshpande2009distributed} studied location-dependent sensing in a distributed setting. For a detailed overview of the literature on sensor coverage, we refer the reader to the survey by Wang \cite{wang2011coverage}.

\subsection{Organization}
 In \cref{sec:prelim}, we introduce the necessary preliminaries, including definitions and supporting lemmas. \cref{sec:weightedPTAS} describes \cref{alg:PTASkmeans}, our proposed PTAS for the weighted \(k\)-means problem. The analysis of this algorithm, and the proof of \cref{thm:ptaskmeans}, is presented in \cref{sec:analysisalgo1}. In \cref{sec:sensor}, we discuss the sensor coverage problem in detail and prove \cref{thm:ptassensor}. In \cref{sec:conclusion}, we present the conclusions of our work, while a few missing proofs are deferred to \cref{sec:appendix}.

\section{Preliminaries}\label{sec:prelim}
Let $P \subset \RR$ be a point set (unweighted). Then, the centroid of $P$, denoted as $G(P)$ is defined as $G(P)=\frac{\sum_{p \in P} p}{|P|}$. Similarly, when $P$ is a weighted point set, the centroid of $P$ 
is defined as 

   $$G(P) =\frac{\sum_{p \in P} w_p \cdot p}{\sum_{p \in P} w_p}.$$
This definition can be easily extended for any $d$-dimensional point set $P$. We now define the weighted $\dset$-sampling, which we will use intensively throughout the paper.

\begin{definition} (Weighted $\dset$-sampling)
    Given a weighted point set $P \subset \RR^d$, the weighted $\dset$-sampling is defined as follows:
\begin{enumerate}

    \item Select the first center at random with probability proportional to $w_p$ and add it to the empty set $\cset'$.

    \item Select the next center with probability $\frac{ w_p \cdot d(p,\cset')^2}{\sum_{p \in P} w_p\cdot d(p,\cset')^2}$ and add it to the set $\cset'$. Here $\cset'$ is the set of centers selected so far.

    \item Repeat step 2, until $k$ centers are obtained in $\cset'$.  
    \end{enumerate}
\end{definition}

Next, we define the optimal $k$-means clustering for the weighted setting. For a weighted point set $P \subset \RR^d$, let $\Delta_k^w(P)$ denote the cost of the optimal weighted $k$-means. We also use $\Delta_k^w(P)$ to denote the optimal weighted clustering wherever required.  We now define the weighted irreducibility property, which helps us to prove \cref{thm:ptaskmeans}. Later, we will show how to remove this property.
 
\begin{definition}(Weighted Irreducibility)
    Given $k$, $\epsilon$, a set of weighted points $P$ with a weight function $w: P \to \RR_{> 0}$, is said to be $(k,\epsilon)$-irreducible if $\Delta_{k-1}^w(P) \ge (1+\epsilon)\Delta_k^w(P).$
    
\end{definition}
Next, we define some notations and a definition with respect to the optimal clustering. Given an optimal $k$-means clustering $\Delta_k^w(P)$, let $O_1, O_2, \dots, O_k$ be the optimal clusters for the point set $P$. Let $c_i$ be the weighted centroid of the points in $O_i$. Let $W_i$ denote the weight of the cluster $O_i$ i.e., $W_i=\sum_{p \in O_i} w_p$. Assume without loss of generality that $W_1 \ge W_2 \ge \dots \ge W_k$. Define $r_i$ to be the average cost incurred by the points in $O_i$, then $$r_i=
\frac{\sum_{p \in O_i}w_p.d(p,c_i)^2}{W_i}.$$

In the following lemma, we define a relationship for the weighted $k$-means with respect to $G(P)$ and for any random point in $\RR^d$. 

\begin{lemma}\label{lem:differentcenter}
    Let $P \subset \RR^d$ be a weighted point set, where each point $p$ is assigned a positive weight $w_p$, and let $W_P=\sum_{p \in P} w_p$. Let $c \in \RR^d$ be any point. Then, 
    
    $$\sum_{p \in P} w_p\cdot d(p,c)^2=\sum_{p \in P} w_p\cdot d(p, G(P))^2 + W_P\cdot d(c,G(P))^2.$$ 
\end{lemma}
\begin{proof}
    As $d(p,c)^2=||p-c||^2$, we can write L.H.S. as $\sum_{p \in P} w_p\cdot ||p-c||^2$. Again, $||p-c||^2 =||p-G(P) + G(P)-c||^2$, which we can write as $$||p-c||^2 =||p-G(P)||^2 + ||G(P)-c||^2 + 2\left< p-G(P),G(P)-c \right>.$$
    
    Considering all the points in $P$ and multiplying both sides by $w_p$ for every $p \in P$ we get

     $$\sum_{p \in P}w_p||p-c||^2 =\sum_{p \in P}w_p||p-G(P)||^2 + \sum_{p \in P}w_p||G(P)-c||^2 + 2\sum_{p \in P}w_p\left< p-G(P),G(P)-c \right>.$$

     As $G(P)$ is a weighted mean, $\sum_{p \in P}w_p (p-G(P))=0$. Thus, the term $$2\sum_{p \in P}w_p\left< p-G(P),G(P)-c \right>=0.$$

     Therefore, we have

    $$\sum_{p \in P}w_p||p-c||^2 =\sum_{p \in P}w_p||p-G(P)||^2 + \sum_{p \in P}w_p||G(P)-c||^2.$$

    As $\sum_{p \in P} w_p= W_P$, we have $\sum_{p \in P}w_p||p-c||^2 =\sum_{p \in P}w_p||p-G(P)||^2 + W_P||G(P)-c||^2$, which proves the lemma.
\end{proof}

The following lemma is from Inaba et al. \cite{inaba1994applications} and shows the power of uniform sampling.

\begin{lemma}(Inaba et al. \cite{inaba1994applications})\label{lem:indepSampling}
Let $S$ be a set of points obtained by independently sampling $M$ points from $P$ (unweighted) with replacement. Then, for $\delta > 0$, with probability at least $1- \delta$ we have, $$\Delta(P,G(S))\le \left(1+ \frac{1}{\delta M}\right)\Delta(P,G(P)).$$
    
\end{lemma}

\section{A Simple PTAS for weighted $k$-means}\label{sec:weightedPTAS}

In this section, we show that the weighted \( k \)-means admits a PTAS by using weighted \( \dset \)-sampling. 
Let \( \mathcal{C} \) denote the set of centers constructed by our algorithm. Initially, \( \mathcal{C} = \emptyset \). The algorithm iteratively builds $\mathcal{C}$ by adding one center at a time until \( |\mathcal{C}| = k \). At each iteration, given the current set of centers \( \mathcal{C} \) with \( |\mathcal{C}| < k \), we perform the following steps:

\begin{enumerate}
    \item \textbf{Weighted \( \dset \)-sampling}: We sample a subset \( S \subseteq P \) of size \( N = O(k/\epsilon^2) \), where the sampling is done using weighted \( \dset \)-sampling with respect to the current centers in \( \mathcal{C} \).

    \item \textbf{Enumerating candidate subsets:} From the sampled set \( S \) of size $N$, we consider all subsets \( T \subseteq S \) of size \( M = O(1/\epsilon) \). For each such subset \( T \), we compute its weighted centroid, denoted as \( G(T) \), and consider adding it to \( \mathcal{C} \) as a potential center. As the number of such subsets in each iteration is \( \binom{N}{M}^k = 2^{\tilde{O}(k/\epsilon)} \), one of these potential centers will be very close to one of the optimal centers. We denote such a center as our candidate center.
\end{enumerate}
 
To simplify the presentation, we preselect a set of \( k \)-tuples of candidate subsets whose centroid is our candidate center. We denote these candidate subsets by \( (s_1, s_2, \dots, s_k) \), where each \( s_i \) indexes a specific subset \( T_i \subseteq S \) of size \( M \). We refer to such \( k \)-tuples as our \emph{candidate tuples}. Our algorithm then constructs a clustering solution for each such candidate tuple by sequentially adding the centers \( G(s_1), G(s_2), \dots, G(s_k) \) to \( \mathcal{C} \). We present the full procedure below as Algorithm 1.

 \begin{algorithm}[H]
 \caption{A Simple PTAS for Weighted $k$-Means Clustering}
 \label{alg:PTASkmeans}
 \SetKwInOut{Input}{Input}\SetKwInOut{Output}{Output}
 \Input{Weighted point set \( P \subset \mathbb{R}^d \), number of clusters \( k \), accuracy parameter \( \epsilon > 0 \)}
 \Output{Set of centers \( \mathcal{C} \subset \mathbb{R}^d \) of size \( k \)}

 Set \( N =  800k/\epsilon^2 \), \( M =  100/\epsilon \), and enumerate all \( k \)-tuples of subsets \( (s_1, \dots, s_k) \in \left[\binom{N}{M}\right]^k \)\\

 Repeat the following for \( 2^k \) independent trials:\\
 \Indp
 Initialize \( \mathcal{C} \leftarrow \emptyset \)\\
 \For{\( i = 1 \) to \( k \)}{
     Sample a subset \( S \subseteq P \) of size \( N \) using weighted \( \dset \)-sampling w.r.t. \( \mathcal{C} \)\\
     Let \( T_i \) be the \( s_i^{\text{th}} \) subset of \( S \) of size \( M \)\\
     Compute the weighted centroid \( G(T_i) \), and update \( \mathcal{C} \leftarrow \mathcal{C} \cup \{G(T_i)\} \)
 }
 Retain the set \( \mathcal{C} \) with the minimum weighted \( k \)-means cost over all iterations
 \end{algorithm}

We now explain our algorithm in detail.
\begin{itemize}
    \item \textbf{Line 1:} Set up parameters for sampling and subset enumeration. The number of sampled points \( N \) and the size of each subset \( M \) are chosen to ensure the approximation guarantee.

    \item \textbf{Line 2:} Repeats the clustering procedure \( 2^k \) times to amplify the success probability of obtaining a good approximation.

    \item \textbf{Lines 4–7:} Iteratively build  \( \mathcal{C} \) (See \cref{fig:algo1}). In each iteration:
    \begin{itemize}
        \item  A set \( S \) of \( N \) points is sampled using weighted \( \dset \)-sampling.

        \item A candidate subset \( T_i=s_i \) of size \( M \) is selected (based on a pre-specified tuple \( (s_1, \dots, s_k) \)).

        \item The weighted centroid of \( T_i \) is added as a new center.
    \end{itemize}
    
\end{itemize}

\begin{figure}
    \centering

    \begin{subfigure}[b]{0.47\textwidth}
        \centering
        \includegraphics[width=\textwidth]{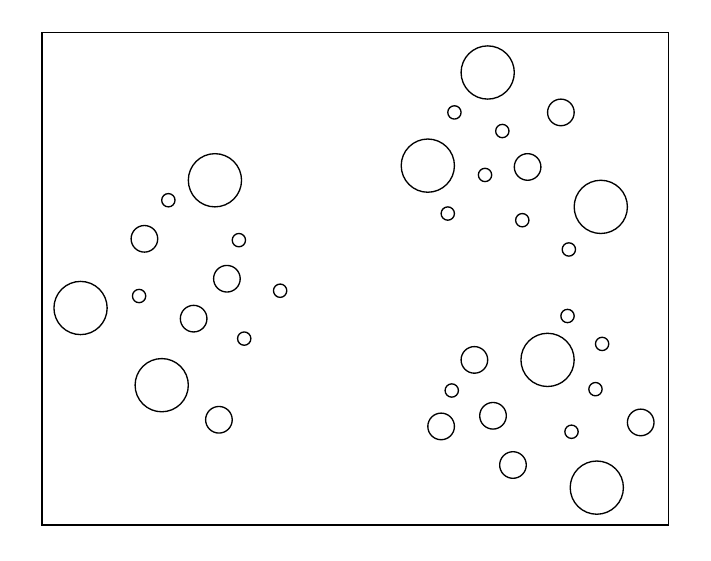}
        \caption{Weighted Point Set $P$ before the start of \cref{alg:PTASkmeans}.}
       \label{fig:k1}
    \end{subfigure}
    \hspace{.5cm}
    \begin{subfigure}[b]{0.47\textwidth}
        \centering
        \includegraphics[width=\textwidth]{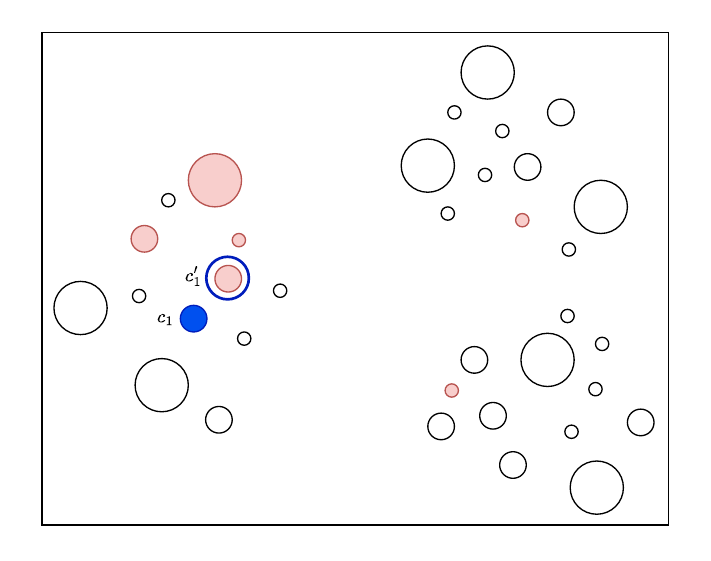}
        \caption{Iteration 1 of  \cref{alg:PTASkmeans}: Pink points are $T_1$, $c_1'=G(T_1)$, and $c_1$ is an optimal center.}
        \label{fig:k2}
    \end{subfigure}
     \hspace{.5cm}
     \begin{subfigure}[b]{0.47\textwidth}
        \centering
        \includegraphics[width=\textwidth]{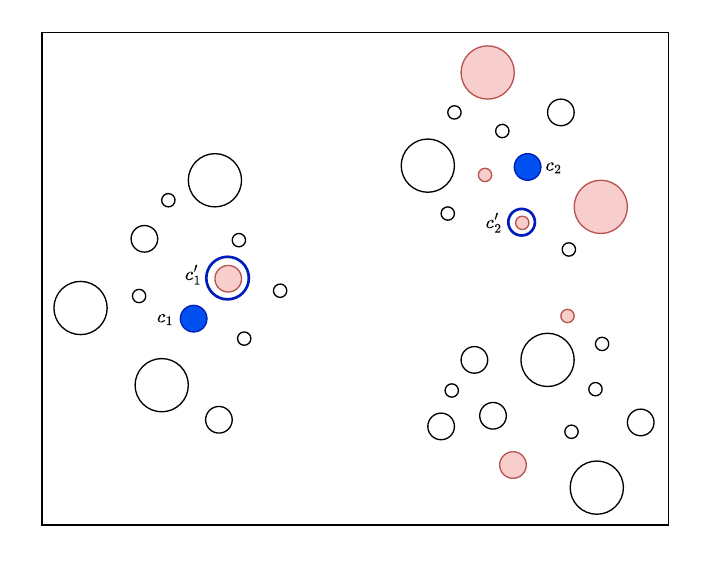}
        \caption{Iteration 2 of  \cref{alg:PTASkmeans}: Pink points are $T_2$, $c_2'=G(T_2)$, and $c_2$ is an optimal center.}
        \label{fig:k3}
    \end{subfigure}
    \hspace{.5cm}
     \begin{subfigure}[b]{0.47\textwidth}
        \centering
        \includegraphics[width=\textwidth]{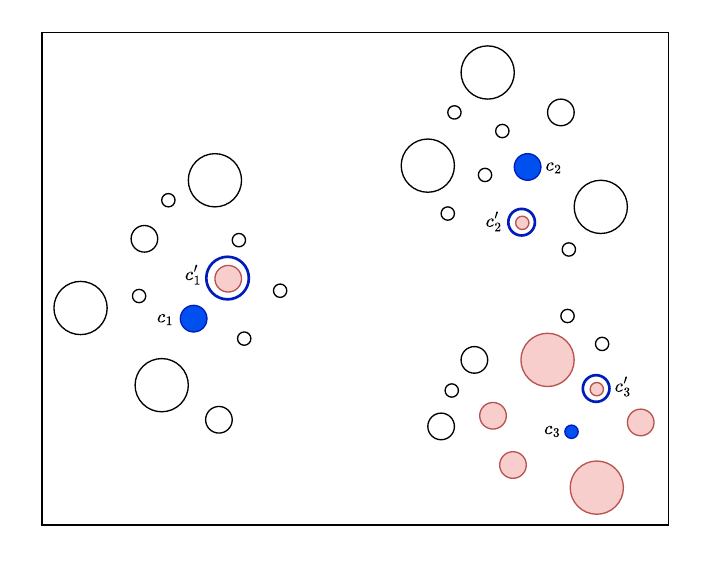}
        \caption{Iteration 3 of  \cref{alg:PTASkmeans}: Pink points are $T_3$, $c_3'=G(T_3)$, and $c_3$ is an optimal center.}
        \label{fig:k4}
    \end{subfigure}
    \hspace{.5cm}
     \begin{subfigure}[b]{0.47\textwidth}
        \centering
        \includegraphics[width=\textwidth]{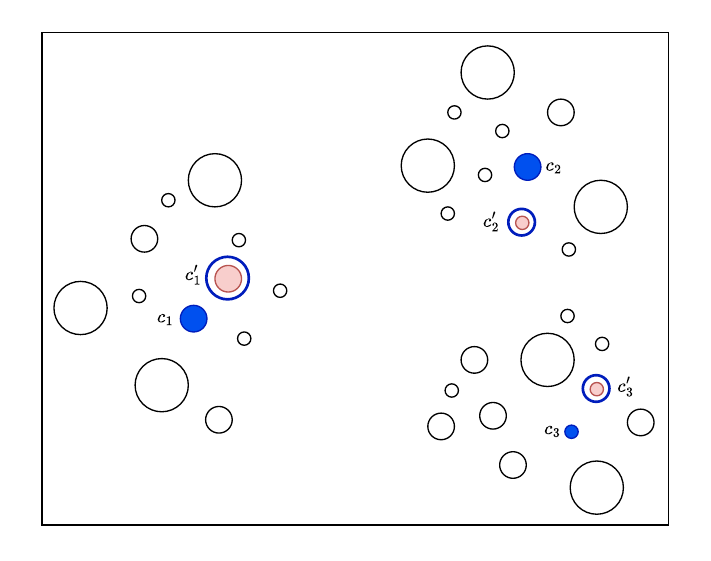}
        \caption{Point set $P$ after the end of \cref{alg:PTASkmeans}.}
        \label{fig:k5}
    \end{subfigure}
    
    \caption{Illustration of  \cref{alg:PTASkmeans}. We have considered $k=3$. In all the above figures, the radius of the points represents their associated weights. The above iteration resembles the \textbf{For Loop}, i.e., lines 4-8. This is exhaustively searched upon a set of $2^{k}$ trials (see line 2), and finally, the best centers are returned (see line 9).}
    \label{fig:algo1}
\end{figure}

Note that the total number of candidate tuples is \( \binom{N}{M}^k = 2^{\tilde{O}(k/\epsilon)} \). For each such tuple, we do centroid computations in time \( \tau nd \) where $\tau$ is some large constant. Repeating this for \( 2^k \) trials, the total running time of  \cref{alg:PTASkmeans} is $\tau'nd\cdot 2^{\tilde{O}(k/\epsilon)}$, where $\tau'$ is a large constant.

\subsection{Analysis of \cref{alg:PTASkmeans}}\label{sec:analysisalgo1}

In this section, we show that the weighted $k$-means algorithm ( \cref{alg:PTASkmeans}) using weighted $\dset$-sampling admits a PTAS. To this end, we first show that any two optimal centers for weighted $k$-means are far apart.

\begin{lemma}\label{lem:distanceoptimalcenter}
    For $1 \le i,j \le k$, and $i \neq j$ we have $d(c_i,c_j)^2 \ge \epsilon(r_i + r_j)$.
\end{lemma}
\begin{proof}
    Assume $i<j$, then $W_i \ge W_j$ by our initial assumption. Assume for contradiction that $d(c_i,c_j)^2 < \epsilon (r_i + r_j)$. Then,
    
\begin{align*}
    \Delta^w(O_i \cup O_j,\{c_i\})&=W_ir_i+ W_jr_j + W_j.d(c_i,c_j)^2 \hspace{13pt}\text{(using Lemma \ref{lem:differentcenter})}\\
 & \le W_ir_i + W_jr_j + W_j.\epsilon.(r_i+ r_j)\\
 & \le (1+\epsilon).W_ir_i + (1+\epsilon)W_jr_j \hspace{20pt}\text{(as $W_i\ge W_j$)}\\
 & \le (1+\epsilon) \Delta^w(O_i \cup O_j, \{c_i,c_j\})\\
\end{align*}

But this contradicts the assumption that $P$ is $(\epsilon,k)$-irreducible. Hence, $d(c_i,c_j)^2 \ge \epsilon (r_i + r_j)$.
    
\end{proof}

To establish that \cref{alg:PTASkmeans} admits a PTAS, it is essential to show that the centers selected by the algorithm remain sufficiently close to the corresponding centers in an optimal solution. To prove this, we maintain an \emph{invariant} (same as \cite{jaiswal2012simple} but for the weighted case): we assume that after \( i-1 \) iterations, the current set of centers \( \mathcal{C}^{(i-1)} = (c_1', c_2', \dots, c_{i-1}') \) constructed by \cref{alg:PTASkmeans} approximates a subset of the optimal centers. Our goal is then to show that the invariant continues to hold after the \( i \)-th center \( c_i' \) is added. 

Formally, let \( \mathcal{C}^{(i-1)} = \{c_1', \dots, c_{i-1}'\} \) denote the centers selected till iteration \( i-1 \). Then, with probability at least \( 1/2^i \), there exist distinct indices \( m_1, m_2, \dots, m_{i-1} \in [k] \) such that for all \( l = (1, \dots, i-1) \), the following holds:

$$ \Delta^w(O_{m_l}, c_l') \leq (1 + \epsilon/20) \cdot \Delta^w(O_{m_l}, c_{m_l}),$$
where \( O_{m_l} \subseteq P \) is the set of points assigned to the optimal center \( c_{m_l} \). We should now prove that this invariant also holds at iteration \( i \), i.e.,

$$ \Delta^w(O_{m_i}, c_i') \leq (1 + \epsilon/20) \cdot \Delta^w(O_{m_i}, c_{m_i}), $$
with probability at least \( 1/2^i \). This will eventually prove that after the $k$-th center is added we have a $(1+\epsilon)$-approximation with probability at least $1/2^k$. In the rest of this section, we show that the invariant holds at iteration $i$.

Let $l \in \{1, \dots, i-1\}$ and define $M' = \{m_1, \dots, m_{i-1}\}$ as the indices of the optimal clusters that have already been approximated by the centers in $\mathcal{C}^{(i-1)}$. In the next lemma, we show that for any index $l \notin M'$ (i.e., $O_i, \dots, O_k$), there is a good chance that points from $O_l$ are sampled in the $i$-th iteration of \cref{alg:PTASkmeans}.

\begin{lemma}\label{lem:cost}
(Same as Lemma 5 in \cite{jaiswal2012simple}, extended to the weighted setting.)

$$
\frac{\sum_{l \notin M'} \Delta^w(O_l, \mathcal{C}^{(i-1)})}{\sum_{l = 1}^k \Delta^w(O_l, \mathcal{C}^{(i-1)})} \ge \frac{\epsilon}{2}.
$$

\end{lemma}

Since the proof closely follows that of Lemma 5 in \cite{jaiswal2012simple}, we include it in ~\cref{sec:appendix} for completeness.

Now, let $m_i$ be the index such that $m_i \notin M'$ and for which $\Delta(O_{m_i},\cset^{(i-1)})$ is maximum. The following corollary, which is a consequence of ~\cref{lem:cost}, shows that there is a good chance of sampling points from the cluster $O_{m_i}$ in the $i$-th iteration.

\begin{corollary}\label{cor:atleastpr1}

$$
\frac{\Delta^w(O_{m_i}, \mathcal{C}^{(i-1)})}{\sum_{l = 1}^k \Delta^w(O_l, \mathcal{C}^{(i-1)})} \ge \frac{\epsilon}{2k}.
$$

\end{corollary}

Combining ~\cref{lem:cost} and ~\cref{cor:atleastpr1}, we conclude that with probability at least $\epsilon / (2k)$, the $i$-th iteration of the algorithm will sample points from an optimal cluster $O_{m_i}$ that has not yet been approximated. This ensures continued progress in covering the remaining optimal clusters.

\begin{observation}\label{obs:atleastpr2}
    Assume that each point $p \in P$ has weight $w_p \ge 1$. For any $l \notin M'$ and point $p \in O_l$, the probability that $p$ is sampled from $O_l$  is at least $\frac{1}{W_l}$.
\end{observation}

Recall that each point \( p \in P \) has an associated weight \( w_p \). One might consider reducing the weighted \(k\)-means problem to the unweighted setting by creating \( w_p \) copies of each point \( p \), treating them as distinct but colocated unweighted points. While this transformation is valid in theory, it does not necessarily preserve the same clustering cost when used with \(\dset\)-sampling for the unweighted case. Specifically, unweighted \(\dset\)-sampling may over or undersample points with high or low weights, leading to suboptimal center selection. In contrast, weighted \(\dset\)-sampling guarantees that the probability of selecting a point is proportional to its weight and squared distance to the nearest center, effectively capturing the intended importance of each point. To this end, to address this issue and facilitate the analysis, in the next two lemmas, we describe how to simulate our weighted \(\dset\)-sampling procedure within an unweighted framework but still obtain a good approximation.

\begin{lemma}\label{lem:prlemma}
Let $P'$ be an unweighted point set. Consider the following experiment:
\begin{itemize}
    \item with probability $\gamma$, a point from $P'$ is selected uniformly at random.
    \item with probability $1 - \gamma$, the experiment results in a null sample.
\end{itemize}
Repeat this experiment $\ell = \frac{400}{\gamma \varepsilon}$ times, obtaining a multiset $T$. Then, with probability at least $3/4$, $T$ contains a subset $U$ of size $\frac{100}{\varepsilon}$ such that:
$$
\Delta(P', G(U)) \leq \left(1 + \frac{\varepsilon}{20}\right) \cdot \Delta(P',G(P')).$$   
\end{lemma}
 
\begin{proof}
    We perform $\ell = \frac{400}{\gamma \varepsilon}$ independent experiments as described below.

Let $T$ be the random variable denoting the number of successful (non-null) draws of the experiment. Then:
$$ \mathbb{E}[T] = \ell \cdot \gamma = \frac{400}{\varepsilon}.$$

By applying the Chernoff bound to $T$ (which is a sum of independent Bernoulli trials):

$$\mathbb{P}\left(T < \frac{1}{4}\mathbb{E}[T]\right) \leq e^{\left(-\mathbb{E}[T].(3/4)^2/2 \right)} < \frac{1}{4}.$$
Thus, with a probability of at least $3/4$, we have:

$$ |T| \geq \frac{100}{\varepsilon}.$$

But the expected value of $|T|$ is $\frac{400}{\epsilon}$, so with probability at least $3/4$, T contains a subset  $|U| \ge \frac{100}{\epsilon}$. Now, we prove the statement by applying \cref{lem:indepSampling}.
\end{proof}

Let $S^{(i)}$ be the sample of size $N$ in the $i$-th iteration. We now show that the invariant also holds for $\cset^i$.

\begin{lemma}\label{lem:main}
    With probability at least 3/4, there exists a subset $T^{(i)}$ of $S^{(i)}$ of size $\frac{100}{\epsilon}$ such that

$$\Delta^w(O_{m_i}, G(T^{(i)})) \le \left( 1 + \frac{\epsilon}{20} \right) \Delta^w_1(O_{m_i}).$$
\end{lemma}
\begin{proof}
    The set $S^{(i)}$ contains $N=\frac{800k}{\epsilon^2}$ samples of $P$ and we are interested in $S^{(i)}\cap O_{m_i}$. Let $Y_1, \dots, Y_N$ be $N$ independent random variables such that for $1\le t \le N$, $Y_t$ is obtained by sampling an element of $P$ using weighted $\dset$-sampling w.r.t $\cset^{(i-1)}$. If the sampled element belongs to $O_{m_i}$, assign it to $Y_t$ or else discard it.

    Let $\gamma=\frac{\epsilon}{2k}$, then \cref{cor:atleastpr1} and \cref{obs:atleastpr2} implies that $Y_t$ is assigned a particular element of $O_{m_i}$ with probability at least $\frac{\gamma}{W_{m_i}}$. We now show how to use \cref{lem:prlemma}.

    For any element $p \in O_{m_i}$, let the probability of $p$ being sampled using weighted $\dset$-sampling be $\frac{\lambda(p)}{W_{m_i}}$. Also note that for all $p \in O_{m_i}$, $\lambda(p) \ge \gamma$. A way to sample a random variable $T_t$ in \cref{lem:prlemma} is as follows: First sample using $Y_t$. If $Y_t$ is not assigned any element, then $T_t$ is null. Otherwise, $Y_t$ is assigned an element $p \in O_{m_i}$, so we define $T_t=p$ with probability $\frac{\gamma}{\lambda(p)}$, null otherwise.

    Observe that with probability $\gamma$, $T_t$ is a uniform sample of $O_{m_i}$ and null with probability $1- \gamma$. Therefore, mimicking the \cref{lem:prlemma}. Also, observe that the sampled elements of $Y_1, \dots, Y_N$ are always a superset of $T_1, \dots, T_N$. We now use  \cref{lem:prlemma} to complete the proof. 
\end{proof}

We repeat \cref{alg:PTASkmeans} for $2^k$ times to get $(1+\epsilon)$-approximation with high probability for weighted $k$-means that runs in time $nd\cdot 2^{\tilde{O}(k/\epsilon)}$. Lastly, we now remove the $(k,\epsilon)$-irreducibility assumption using Theorem 1 and Theorem 2 of \cite{jaiswal2012simple} to get our final running time. We state the theorems here for completeness.

\begin{theorem}(Theorem 1 of \cite{jaiswal2012simple})\label{thm:thm1}
    If a given point set is $\left(k,\frac{\epsilon}{(1+\epsilon/2)\cdot k} \right)$-irreducible, then there is an algorithm that gives a $\left(1+ \frac{\epsilon}{(1+\epsilon/2)\cdot k} \right)$-approximation to the $k$-means objective and that runs in time $O(nd \cdot 2^{\tilde{O}(k^2/\epsilon)})$.
\end{theorem}

\begin{theorem}(Theorem 2 of \cite{jaiswal2012simple})\label{thm:thm2}
    There is an algorithm that runs in time $O(nd \cdot 2^{\tilde{O}(k^2/\epsilon)})$ and gives a $(1 + \epsilon)$-approximation to the $k$-means objective.
\end{theorem}

Using \cref{alg:PTASkmeans}, \cref{thm:thm1}, and \cref{thm:thm2}, we get the following result.

\begin{theorem}\label{thm:mainptaskmeans}
    Let $P \subset \RR^d$ be a point set such that each point $p \in P$ has positive finite weight $w_p\ge 1$. Then there is a set $\cset=(c_1',c_2',\dots,c_k')$ of $k$ centers computed in time $O(nd \cdot 2^{\tilde{O}(k^2/\epsilon)})$ such that $\Delta^w(P,\cset)\le (1+\epsilon)\Delta_k^w(P)$.
\end{theorem}

As an application of  \cref{thm:mainptaskmeans}, we now show an improved approximation for \emph{sensor coverage problem} from $O(\log k)$ \cite{deshpande2014guaranteed} to $(1+\epsilon)$ (with running time $O(nd \cdot 2^{\tilde{O}(k^2/\epsilon)})$).

\section{A Simple PTAS for Sensor Coverage Problem}\label{sec:sensor}
In this section, we formally define the sensor coverage problem and establish its connection to the weighted $k$-means problem \cite{deshpande2014guaranteed}. Leveraging  \cref{thm:mainptaskmeans}, we then derive a PTAS for the sensor coverage problem.

In the sensor coverage problem, we are given a convex region $Z \subseteq \mathbb{R}^2$. Each point $z \in Z$ has an associated nonnegative value $\phi(z)$, which represents the importance of covering that location. The function $\phi$ is normalized so that the total importance in the region is equal to one, that is, $\int_Z \phi(z)\, dz = 1$.

We aim to place $k$ sensors in the region, with their positions denoted by $\mathcal{C} = (c_1, \dots, c_k) \subset Z$. The effectiveness of a sensor at covering a location $z$ decreases as the distance between $z$ and the sensor increases. Specifically, the sensing performance at point $z$ with respect to a sensor located at $c_i$ degrades proportionally to the square of the Euclidean distance $\|z - c_i\|^2$. For a fixed sensor locations $(c_1,c_2,\dots, c_k)$, the sensor performance induces a \emph{Voronoi partition} $V(\cset)=\{V_1,\dots,V_k\}$ of $Z$, where,

$$V_i=\{z \hspace{.1cm} | \hspace{.1cm} ||z-c_i||^2 \le ||z-c_j||^2 \hspace{.1cm} \forall j \neq i\}.$$
This implies that the sensor located at $c_i$ covers all the points in $V_i$. Therefore, for the sensor locations $\cset$, the coverage cost, denoted as $H(\cset)$, is defined as:

$$H(\cset)=\sum_{i=1}^k \int_{V_i} ||z-c_i||^2 \phi(z) dz.$$
The optimal sensor coverage problem is to find a sensor location $\cset^o=(c_1^o,\dots,c_k^o)$ for which the coverage cost is minimized, i.e.,

$$\cset^o= \arg \min_{\cset} H(\cset).$$

Deshpande \cite{deshpande2014guaranteed} showed that the sensor covering problem can be posed as a weighted $k$-means problem. To this end, we show how to reduce the sensor covering problem to the weighted $k$-means problem, following \cite{deshpande2014guaranteed}.

Given a convex region $Z$, we discretize $Z$ by overlaying a uniform square grid of size $\epsilon \times \epsilon$. The center of each resulting \emph{cell} in this grid is treated as a potential sensor location. Let the cells intersecting $Z$ be denoted by $G_1, G_2, \dots, G_n$. Among these, some are full square cells lying entirely within $Z$, while others, lying along the boundary, may form convex polygonal shapes (see \cref{fig:sensor}). For each cell $G_i$, we define its \emph{weight} as:

$$w_i = \int_{G_i} \phi(z)\, dz,$$

and its \emph{center of mass}, denoted by $x_i$, as:

$$x_i = \frac{1}{w_i} \int_{G_i} z\, \phi(z)\, dz.$$

\begin{figure}[ht]
  \centering
  \includegraphics[scale=0.2]{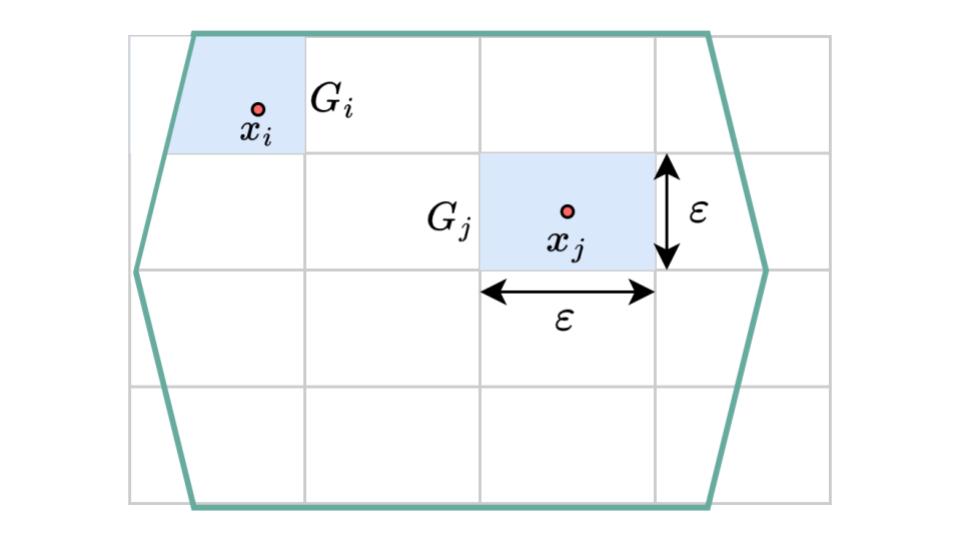}
  \captionsetup{justification=centering}
  \caption{This figure depicts a convex region $Z$ which is superimposed on a grid of size $\epsilon \times\epsilon$.}
  \label{fig:sensor}
\end{figure}

Let $X = (x_1, x_2, \dots, x_n)$ be the set of candidate locations for placing the $k$ sensors. Each $x_i$ is a point in $\mathbb{R}^2$ and is associated with a nonnegative weight $w_i$. The set $X$ can therefore be seen as a weighted point set in the plane. We use $\Delta^w(X, \mathcal{C})$ to denote the weighted $k$-means cost of assigning the points in $X$ to center $\mathcal{C} = (c_1, \dots, c_k)$. Deshpande \cite{deshpande2014guaranteed} established the following connection between weighted $k$-means and the sensor coverage problem.

\begin{theorem}(Theorem 2 of \cite{deshpande2014guaranteed})\label{thm:deshpande}
    $H(\cset)=\Delta^w(X,\cset) + \sum_{i} J_{G_i|x_i}$, where $J_{G_i|x_i}$ is the moment of inertia of cell $G_i$ about its center of mass $x_i$. 
\end{theorem}

Deshpande \cite{deshpande2014guaranteed} showed that when the grid is sufficiently fine (i.e., for small enough $\epsilon$), the term $\sum_i J_{G_i|x_i}$ can be bounded by a small constant and thus can be safely ignored in the analysis. As a result, ~\cref{thm:deshpande} implies that solving the sensor coverage problem reduces to solving the weighted $k$-means problem. Using this connection, Deshpande obtained an $O(\log k)$-approximation for the sensor coverage problem by providing an $O(\log k)$-approximation algorithm for the weighted $k$-means problem. To achieve this, Deshpande extended the analysis of Arthur and Vassilvitskii \cite{arthur2006k} 
to the weighted setting using weighted $\dset$-sampling.

Due to the simple connection between the weighted $k$-means problem and the sensor coverage problem, we can now obtain a $(1+\epsilon)$-approximation for sensor coverage by applying weighted $\dset$-sampling. This leads to the following result:

\begin{theorem}\label{thm:mainptassensor}
Let $\mathcal{C} = (c_1', c_2', \dots, c_k')$ be the set of $k$ centers returned by ~\cref{alg:PTASkmeans} to minimize $\Delta^w(X, \mathcal{C})$. Then, $H(\mathcal{C}) \leq (1 + \epsilon) H(\mathcal{C}^o).$

\end{theorem}

\begin{proof}
We begin by invoking ~\cref{thm:deshpande}, which establishes a reduction from the sensor coverage problem to the weighted $k$-means problem. We then apply ~\cref{alg:PTASkmeans} to solve the weighted $k$-means instance and obtain the set of centers $\mathcal{C}$.

According to ~\cref{thm:mainptaskmeans}, ~\cref{alg:PTASkmeans} computes a $(1+\epsilon)$-approximation for the weighted $k$-means problem in time $O(nd \cdot 2^{\tilde{O}(k^2/\epsilon)})$. Therefore, by connection to sensor coverage, this directly yields a $(1+\epsilon)$-approximation for the sensor coverage problem.
\end{proof}

\section{Conclusion}\label{sec:conclusion}

In this work, we presented a simple PTAS for the weighted $k$-means problem that avoids the complexity of coreset constructions. Our algorithm runs in time $nd \cdot 2^{\tilde{O}(k^2 / \epsilon)}$ and is both easy to analyze and straightforward to implement, making it a practical alternative to more involved coreset-based approaches.

Using our weighted $k$-means result, we obtained a PTAS for the sensor coverage problem, improving upon the previous $O(\log k)$-approximation of Deshpande ~\cite{deshpande2014guaranteed}.

\bibliographystyle{alpha}

\begin{thebibliography}{CGK{\etalchar{+}}15}

\bibitem[JKS12]{jaiswal2012simple}
R.~Jaiswal, A.~Kumar, and S.~Sen, 
``A simple D$^2$-sampling based PTAS for k-means and other clustering problems,'' 
in \textit{Computing and Combinatorics: 18th Annual International Conference, COCOON 2012, Sydney, Australia, August 20-22, 2012. Proceedings 18}, 
pp.~13--24, 2012, Springer.

\bibitem[Des14]{deshpande2014guaranteed}
A.~Deshpande, 
``Guaranteed sensor coverage with the weighted-$D^2$ sampling,'' 
\textit{arXiv preprint arXiv:1412.0301}, 2014.

\bibitem[IKI94]{inaba1994applications}
M.~Inaba, N.~Katoh, and H.~Imai, 
``Applications of weighted Voronoi diagrams and randomization to variance-based k-clustering,'' 
in \textit{Proc. 10th Annual Symposium on Computational Geometry}, 
pp.~332--339, 1994.

\bibitem[AV06a]{arthur2006k}
D.~Arthur and S.~Vassilvitskii, 
``k-means++: The advantages of careful seeding,'' 
Stanford University, Tech. Rep., 2006.

\bibitem[Das08]{dasgupta2008hardness}
S.~Dasgupta, 
``The hardness of k-means clustering,'' 
2008.

\bibitem[Llo82]{lloyd1982least}
S.~Lloyd, 
``Least squares quantization in PCM,'' 
\textit{IEEE Transactions on Information Theory}, vol.~28, no.~2, pp.~129--137, 1982.

\bibitem[AV06b]{arthur2006slow}
D.~Arthur and S.~Vassilvitskii, 
``How slow is the k-means method?,'' 
in \textit{Proc. 22nd Annual Symposium on Computational Geometry}, 
pp.~144--153, 2006.

\bibitem[Mat00]{matouvsek2000approximate}
J.~Matou{\v{s}}ek, 
``On approximate geometric k-clustering,'' 
\textit{Discrete \& Computational Geometry}, vol.~24, no.~1, pp.~61--84, 2000.

\bibitem[KMN$^{+}$02]{kanungo2002local}
T.~Kanungo, D.~M. Mount, N.~S. Netanyahu, C.~D. Piatko, R.~Silverman, and A.~Y. Wu, 
``A local search approximation algorithm for k-means clustering,'' 
in \textit{Proc. 18th Annual Symposium on Computational Geometry}, 
pp.~10--18, 2002.

\bibitem[KSS04]{kumar2004simple}
A.~Kumar, Y.~Sabharwal, and S.~Sen, 
``A simple linear time (1+ $\epsilon$)-approximation algorithm for k-means clustering in any dimensions,'' 
in \textit{Proc. 45th Annual IEEE Symposium on Foundations of Computer Science}, IEEE, 2004.

\bibitem[CMKB04]{cortes2004coverage}
J.~Cortes, S.~Martinez, T.~Karatas, and F.~Bullo, 
``Coverage control for mobile sensing networks,'' 
\textit{IEEE Transactions on Robotics and Automation}, vol.~20, no.~2, pp.~243--255, 2004.

\bibitem[KSS10]{kumar2010linear}
A.~Kumar, Y.~Sabharwal, and S.~Sen, 
``Linear-time approximation schemes for clustering problems in any dimensions,'' 
\textit{Journal of the ACM}, vol.~57, no.~2, pp.~1--32, 2010.

\bibitem[HS05]{har2005fast}
S.~Har-Peled and B.~Sadri, 
``How fast is the k-means method?,'' 
\textit{Algorithmica}, vol.~41, pp.~185--202, 2005.

\bibitem[Che06]{chen2006k}
K.~Chen, 
``On k-median clustering in high dimensions,'' 
in \textit{Proc. 17th Annual ACM-SIAM Symposium on Discrete Algorithms}, 
pp.~1177--1185, 2006.

\bibitem[HM04]{har2004coresets}
S.~Har-Peled and S.~Mazumdar, 
``On coresets for k-means and k-median clustering,'' 
in \textit{Proc. 36th Annual ACM Symposium on Theory of Computing}, 
pp.~291--300, 2004.

\bibitem[FMS07]{feldman2007ptas}
D.~Feldman, M.~Monemizadeh, and C.~Sohler, 
``A PTAS for k-means clustering based on weak coresets,'' 
in \textit{Proc. 23rd Annual Symposium on Computational Geometry}, 
pp.~11--18, 2007.

\bibitem[CGL$^{+}$22]{cohen2022improved}
V.~Cohen-Addad, K.~G. Larsen, D.~Saulpic, C.~Schwiegelshohn, and O.~A. Sheikh-Omar, 
``Improved coresets for Euclidean k-means,'' 
\textit{Advances in Neural Information Processing Systems}, vol.~35, pp.~2679--2694, 2022.

\bibitem[ASI20]{ahmed2020k}
M.~Ahmed, R.~Seraj, and S.~M.~S. Islam, 
``The k-means algorithm: A comprehensive survey and performance evaluation,'' 
\textit{Electronics}, vol.~9, no.~8, p.~1295, 2020.

\bibitem[Fel20]{feldman2020core}
D.~Feldman, 
``Core-sets: Updated survey,'' 
in \textit{Sampling Techniques for Supervised or Unsupervised Tasks}, pp.~23--44, Springer, 2020.

\bibitem[MV22]{mitzenmacher2022algorithms}
M.~Mitzenmacher and S.~Vassilvitskii, 
``Algorithms with predictions,'' 
\textit{Communications of the ACM}, vol.~65, no.~7, pp.~33--35, 2022.

\bibitem[EFS$^{+}$21]{ergun2021learning}
J.~C. Ergun, Z.~Feng, S.~Silwal, D.~P. Woodruff, and S.~Zhou, 
``Learning-augmented k-means clustering,'' 
\textit{arXiv preprint arXiv:2110.14094}, 2021.

\bibitem[NCN22]{nguyen2022improved}
T.~Nguyen, A.~Chaturvedi, and H.~L. Nguyen, 
``Improved learning-augmented algorithms for k-means and k-medians clustering,'' 
\textit{arXiv preprint arXiv:2210.17028}, 2022.

\bibitem[HFH$^{+}$]{huangnew}
J.~Huang, Q.~Feng, Z.~Huang, Z.~Zhang, J.~Xu, and J.~Wang, 
``New algorithms for the learning-augmented k-means problem,'' 
in \textit{Proc. 13th International Conference on Learning Representations}.

\bibitem[KKS05]{kerdprasop2005weighted}
K.~Kerdprasop, N.~Kerdprasop, and P.~Sattayatham, 
``Weighted k-means for density-biased clustering,'' 
in \textit{International Conference on Data Warehousing and Knowledge Discovery}, 
pp.~488--497, 2005, Springer.

\bibitem[BJN12]{baswade2012comparative}
A.~M. Baswade, K.~D. Joshi, and P.~S. Nalwade, 
``A comparative study of k-means and weighted k-means for clustering,'' 
\textit{International Journal of Engineering Research \& Technology}, vol.~1, no.~10, 2012.

\bibitem[MS03]{modha2003feature}
D.~S. Modha and W.~S. Spangler, 
``Feature weighting in k-means clustering,'' 
\textit{Machine Learning}, vol.~52, pp.~217--237, 2003.

\bibitem[DeA16]{de2016survey}
R.~C. De~Amorim, 
``A survey on feature weighting based k-means algorithms,'' 
\textit{Journal of Classification}, vol.~33, pp.~210--242, 2016.

\bibitem[OS97]{okabe1997locational}
A.~Okabe and A.~Suzuki, 
``Locational optimization problems solved through Voronoi diagrams,'' 
\textit{European Journal of Operational Research}, vol.~98, no.~3, pp.~445--456, 1997.

\bibitem[SMR06]{schwager2006distributed}
M.~Schwager, J.~McLurkin, and D.~Rus, 
``Distributed coverage control with sensory feedback for networked robots,'' 
in \textit{Robotics: Science and Systems}, pp.~49--56, 2006.

\bibitem[DPRS09]{deshpande2009distributed}
A.~Deshpande, S.~Poduri, D.~Rus, and G.~S. Sukhatme, 
``Distributed coverage control for mobile sensors with location-dependent sensing models,'' 
in \textit{Proc. 2009 IEEE International Conference on Robotics and Automation}, 
pp.~2344--2349, IEEE, 2009.

\bibitem[Wan11]{wang2011coverage}
B.~Wang, 
``Coverage problems in sensor networks: A survey,'' 
\textit{ACM Computing Surveys}, vol.~43, no.~4, pp.~1--53, 2011.

\bibitem[CSS21]{cohen2021new}
V.~Cohen-Addad, D.~Saulpic, and C.~Schwiegelshohn, 
``A new coreset framework for clustering,'' 
in \textit{Proc. 53rd Annual ACM SIGACT Symposium on Theory of Computing}, 
pp.~169--182, 2021.

\bibitem[HLW24]{huang2024optimal}
L.~Huang, J.~Li, and X.~Wu, 
``On optimal coreset construction for Euclidean $(k, z)$-clustering,'' 
in \textit{Proc. 56th Annual ACM Symposium on Theory of Computing}, 
pp.~1594--1604, 2024.

\end{thebibliography}
\newcommand{\etalchar}[1]{$^{#1}$}

\appendix

\section{Proof of \cref{lem:cost}}\label{sec:appendix}

For contradiction, assume the statement is true. Then,

\begin{align*}
    \Delta^w(P,\cset^{(i-1)}) &= \sum_{m \in M' } \Delta^w(O_m,\cset^{(i-1)}) +  \sum_{m \neq M' } \Delta^w(O_m,\cset^{(i-1)})\\
    &< \sum_{m \in M' } \Delta^w(O_m,\cset^{(i-1)}) + \frac{\epsilon/2}{1-\epsilon/2}.\sum_{m \in M' } \Delta^w(O_m,\cset^{(i-1)})\\
    &= \frac{1}{1-\epsilon/2}.\sum_{m \in M' } \Delta^w(O_m,\cset^{(i-1)})\\
    & \le \frac{1+ \epsilon/20}{1-\epsilon/2}. \sum_{m \in M'} \Delta^w_1(O_m)  \hspace{1cm}\text{Using invariant for $\cset^{(i-1))}$}\\
    & \le (1+\epsilon).\sum_{m \in M'} \Delta^w_1(O_m) \le (1 + \epsilon) \sum_{m \in [k]} \Delta^w_1(O_m)
\end{align*}

But this contradicts that $P$ is $(k,\epsilon)$-irreducible.

\end{document}